\documentclass[onecolumn,english,aps,pra,floatfix,tightenlines,superscriptaddress,nofootinbib,twocolumn]{revtex4-1}
\usepackage[T1]{fontenc}
\usepackage[latin1]{inputenc}
\usepackage{amsmath}
\usepackage{amsthm}
\usepackage{amssymb}
\usepackage{graphicx}
\usepackage{esint}
\usepackage{comment}
\usepackage[up]{subfigure}

\usepackage{hyperref}
\hypersetup{
    colorlinks=true,       
    linkcolor=cyan,          
    citecolor=magenta,        
    filecolor=magenta,      
    urlcolor=cyan,           
    runcolor=cyan
}

\newcommand{\eea}{\end{eqnarray}}
\newcommand{\bea}{\begin{eqnarray}}
\newcommand{\bes}{\begin{subequations}}
\newcommand{\ees}{\end{subequations}}
\newcommand{\beq}{\begin{equation}}
\newcommand{\eeq}{\end{equation}}

\newcommand{\DL}[1]{\textcolor{red}{[DL: #1]}}

\newcommand{\ignore}[1]{}

\makeatletter

\pdfpageheight\paperheight
\pdfpagewidth\paperwidth

\theoremstyle{plain}
\newtheorem{thm}{\protect\theoremname}
\theoremstyle{plain}
\newtheorem{prop}[thm]{\protect\propositionname}
\theoremstyle{remark}

\usepackage{times}

\newcommand{\hc}{\mathrm{h.c.}}

\DeclareMathOperator{\Tr}{Tr}

\newcommand{\1}{\leavevmode{\rm 1\ifmmode\mkern  -4.8mu\else\kern -.3em\fi I}}

\makeatother

\usepackage{babel}
\providecommand{\propositionname}{Proposition}
\providecommand{\remarkname}{Remark}
\providecommand{\theoremname}{Theorem}

\begin{document}

\title{Error Reduction in Quantum Annealing using Boundary Cancellation: Only the End Matters}

\author{Lorenzo Campos Venuti}
\affiliation{Department of Physics and Astronomy, University of Southern California, Los
  Angeles, California 90089, USA}
\affiliation{Center for Quantum Information Science \& Technology, University of Southern
  California, Los Angeles, California 90089, USA}

\author{Daniel A. Lidar}
\affiliation{Department of Physics and Astronomy, University of Southern California, Los
  Angeles, California 90089, USA}
\affiliation{Center for Quantum Information Science \& Technology, University of Southern
  California, Los Angeles, California 90089, USA}
\affiliation{Department of Electrical Engineering, University of Southern California,
Los Angeles, CA 90089, USA}
\affiliation{Department of Chemistry, University of Southern California, Los Angeles,
CA 90089, USA}

\date{\today}
\begin{abstract}
The adiabatic theorem of quantum mechanics states that the error between
an instantaneous eigenstate of a time-dependent Hamiltonian and the
state given by quantum evolution of duration $\tau$ is upper bounded
by $C/\tau$ for some positive constant $C$. It has been known for
decades that this error can be reduced to $C_{k}/\tau^{k+1}$ if the
Hamiltonian has vanishing derivatives up to order $k$ at the beginning
and end of the evolution. Here we extend this result to open systems
described by a time-dependent Liouvillian superoperator. We find that
the same results holds provided the Liouvillian has vanishing derivatives
up to order $k$ only at the end of the evolution. This asymmetry
is ascribable to the arrow of time inherent in open system evolution.
We further investigate whether it is possible to satisfy the required
assumptions by controlling only the system, as required for realistic
implementations. Surprisingly, we find the answer to be affirmative.
We establish this rigorously in the setting of the Davies-Lindblad
adiabatic master equation, and numerically in the setting of two different
time-dependent Redfield-type master equations we derive. The results
are shown to be stable with respect to imperfections in the preparation.
Finally, we prove that the results hold also in a fully Hamiltonian
model. 
\end{abstract}
\maketitle

\section{Introduction}
\label{sec:intro}

Quantum annealing and adiabatic quantum computation
are promising candidates in the search for quantum-enhanced information processing~\cite{RevModPhys.80.1061,Albash-Lidar:RMP}. Both can be viewed as adiabatic state preparation protocols~\cite{AharonovTa-Shma:03}, where the target state is typically the solution to a computational problem such as optimization, or a state from a distribution that one wishes to sample from.
In the simplest scenario one simply interpolates linearly from a 
Hamiltonian with an easily prepared ground state to a target Hamiltonian whose ground state encodes the
solution to a given computational problem~\cite{kadowaki_quantum_1998,farhi_quantum_2000}. In this closed-system, ideal, zero-temperature
case the time needed to reach an accuracy $\epsilon$ scales inversely
proportional to $\epsilon$ and to some power of the inverse minimum
gap \cite{jansen_bounds_2007}. A similar result holds also for preparing
thermal equilibrium (Gibbs) states in a realistic, open-system setting, in that the time
needed to reach an accuracy $\epsilon$ from the Gibbs state is still
inversely proportional to $\epsilon$~\cite{venuti_adiabaticity_2016}. The dependence on the problem
size is, however, more complicated in the open system setting (where the dynamics is generated by a Liouvillian rather than a Hamiltonian), primarily because
in this case not only the Hamiltonian gap but also the Liouvillian
gap plays a role \cite{venuti_adiabaticity_2016,venuti_relaxation_2017}. Moreover, in the
open-system setting one does not expect a speedup with respect to classical
preparation algorithms if the temperature is sufficiently high. Arguments
to substantiate this statement on general grounds were given in \cite{venuti_relaxation_2017},
and it can be rigorously proven in certain specific cases \cite{temme_how_2015}. 

Given this state of affairs and the intense efforts to realize adiabatic quantum state preparation experimentally with an eye towards quantum speedups~\cite{speedup}, it is critical to find general protocols for preparing quantum states using the adiabatic approach
that offer provable advantages over naive protocols such as linear interpolation. It is well known that protocols that slow down near the quantum critical point are beneficial~\cite{PhysRevLett.103.080502} and sometimes even necessary for achieving a quantum speedup~\cite{Roland:2002ul}.
In the closed system setting there exists another general method that
allows for a reduction of the adiabatic error (the distance between
the desired state and the state that has been actually prepared) from
$\propto1/\tau$ where $\tau$ is the total evolution time (annealing,
or preparation time) to $\propto1/\tau^{k+1}$ where the exponent
$q\ge1$ can be made arbitrarily large \cite{garrido_degree_1962}.
In fact, it is even possible to achieve an adiabatic error exponentially
small in $\tau$ which in turns means an annealing time logarithmic
in $1/\epsilon$ \cite{nenciu_linear_1993,hagedorn_elementary_2002,lidar_adiabatic_2009,Cheung:2011aa,Wiebe:12,Ge:2015wo}.
This method simply requires the time-dependent Hamiltonian to have
vanishing derivatives up to order $k$ at the initial and final time.
In other words, the schedule should be sufficiently flat at the beginning
and at the end of the anneal. 

Here we generalize the boundary cancellation method to the open system
case, where the dynamics are generated by a time-dependent Liouvillian
$\mathcal{L}_{\tau}(t)$. We find that, in contrast to the familiar
closed-system result, an asymmetry with respect to the Hamiltonian
case appears, in that the same result holds provided the \emph{Liouvillian}
has vanishing derivatives only at the \emph{end} of the evolution.
The origin of this asymmetry traces back to the time-asymmetry of
non-unitary evolution, which admits an arrow of time. We then consider
whether it is possible to satisfy the required condition on the Liouvillian
by controlling \emph{only the system} (and not the bath degrees of
freedom), as would be required, e.g., for applications using experimental
quantum annealers. To this end we first consider the time-dependent
Davies-Lindblad type adiabatic master equation~\cite{albash_quantum_2012,Albash:2015nx}
and, encouragingly, find the answer to be positive, in that it suffices
to enforce that the system Hamiltonian alone has vanishing derivatives
in order for the adiabatic error to be upper bounded by $C_{k}/\tau^{k+1}$.
This result requires complete positivity, a condition that is automatically
satisfied in this case. To check the robustness of this result to
different levels of approximations we also derive time-dependent adiabatic
as well as non-adiabatic Redfield-type master equations. These master
equations have not appeared previously, to the best of our knowledge,
and should be of independent interest. 
In the  Redfield case the lack of a complete positivity guarantee
prevents us from satisfying the assumptions required for our previous
result to hold. However, our numerical simulations confirm that, in
a parameter regime for which positivity  is
satisfied, enforcing vanishing derivatives of  the system Hamiltonian alone results in a greatly diminished adiabatic error.

The important question that remains to be answered concerns the scaling
of the adiabatic error with other parameters such as the number of
qubits $N$ or the temperature. The scaling with $N$ is not yet fully
understood in the closed system case, and the situation for open systems
with or without boundary cancellation is even more complicated. 
We show that within the range of parameters tested in our numerical simulations, boundary cancellation 
provides an advantage for all annealing times,
with an advantage that is more pronounced in the large $\tau$ region.

In Section~\ref{sec:BC} we formulate boundary cancellation in terms
of condition on the derivatives of the Liouvillian at the end of the
evolution. To do so we first give the general setting for the theory
in terms of trace and hermitian preserving superoperators and describe
a useful adiabatic expansion in powers of the evolution time $\tau$.
We also provide a stability and time-scale analysis. In Section~\ref{sec:application}
we apply the general theory in the setting of various master equations
derived from first-principles, and show both analytically and numerically
that---remarkably---boundary cancellation works by controlling only
the system Hamiltonian. We conclude in Section~\ref{sec:conc}, and
provide additional technical details in the Appendix.

\section{Boundary Cancellation in Open Systems}
\label{sec:BC}

\subsection{Setup}
For simplicity we consider a system with a finite-dimensional dimensional
Hilbert space $\mathcal{H\simeq\mathbb{C}}^{n}$ and let $L(\mathcal{H})$
be the algebra of linear operators over it. We fix the norm on $L(\mathcal{H})$ to be the trace norm: $\left\Vert X\right\Vert _{1}:=\Tr \sqrt{X^\dagger X}$ for $X\in L(\mathcal{H})$.  Let a time-dependent Liouvillean
superoperator $\mathcal{L}_{\tau}(t)$ acting on $L(\mathcal{H})$
be given. The evolution of the system (characterized by the quantum
state $\rho_{\tau}(t)$ at time $t$) is described by a time-dependent
linear differential equation 
\begin{equation}
\frac{\partial\rho_{\tau}(t)}{\partial t}=\mathcal{L}_{\tau}(t)\rho_{\tau}(t).
\label{eq:ME}
\end{equation}
In some cases $\mathcal{L}_{\tau}(t)$ depends on $t$ only through
the rescaled time variable $s=t/\tau$ and we define $\mathcal{L}(s):=\mathcal{L}_{\tau}(t)$.
The time-scale $\tau$ is the total evolution (``anneal'') time.
Note that $\mathcal{L}(s)$ depends on $\tau$ if $\mathcal{L}_{\tau}$
is not simply a function of $t/\tau$. Switching to the variable $s$
and defining $\rho(s):=\rho_{\tau}(t)$ allows us to rewrite Eq.~\eqref{eq:ME} in the form
\begin{equation}
\dot{\rho}(s)=\tau\mathcal{L}(s)\rho(s),
\label{eq:ME_rescaled}
\end{equation}
where henceforth the dot denotes differentiation with respect to $s$.
Evolution up to time $\tau$ thus becomes evolution up to $s=1$.
For convenience we also define $\zeta=1/\tau$. 
Below we use both the time $t$ and rescaled time $s$, whichever
is more convenient. We use $\mathcal{L}(s)$ to denote a linear, trace
preserving and hermitian preserving (TPHP) superoperator for all $s\ge0$.
Occasionally we will assume more, namely that $\mathcal{L}(s)$ generates
a contraction semigroup with respect to the induced norm $\left\Vert \mathcal{T}\right\Vert _{1,1}:=\sup_{x\neq0}\left\Vert \mathcal{T}(x)\right\Vert _{1}/\left\Vert x\right\Vert _{1}$,
meaning that $\left\Vert e^{t\mathcal{L}(s)}\right\Vert _{1,1}\le1$
for all $s,t\ge0$. This includes generators that are
in Lindblad form for all $s\ge0$. 

The propagator or evolution operator is the solution of the following differential equation:
\beq
\partial_{s}\mathcal{E}(s,s')=\tau\mathcal{L}(s)\mathcal{E}(s,s'),\quad\mathcal{E}(s,s)=\1.
\eeq
The adiabatic approximation or expansion refers to the solution of Eq.~\eqref{eq:ME_rescaled} when $\tau\to\infty$. When a gap condition
is satisfied the adiabatic expansion is an expansion in powers of
$\tau^{-1}$. By gap condition we mean that the eigenvalue being followed is separated from the rest of the spectrum by a finite gap uniformly for all $s$ in the evolution window $[0,1]$.
In finite dimensions this is the only possibility if one excludes
level crossings. 

Let $P(s)$ be the eigenprojector of $\mathcal{L}(s)$ with eigenvalue $0$. A $0$ eigenvalue always
exist whenever $\mathcal{L}(s)$ is trace-preserving. Moreover if
$\mathcal{L}(s)$ generates a contraction semi-group, the eigenvalue $0$ does not
have a nilpotent term (these and various other useful facts about Eq.~\eqref{eq:ME_rescaled} are collected in \cite{venuti_adiabaticity_2016}). Let us also denote by $Q(s)=\1-P(s)$ the complementary
projection. For simplicity we assume the system to be finite-dimensional
although all the results still hold in the infinite-dimensional case,
possibly after introducing some extra assumptions. 

\subsection{Adiabatic expansion}
We first provide an adiabatic expansion for the case of a non-degenerate
steady state \textendash{} the corresponding generators are generally called \emph{ergodic}.
This is essentially Theorem 6 of Avron \textit{et al}.~\cite{avron_adiabatic_2012} with
some additional simplifying assumptions. 
\begin{prop}
\label{prop:adiabatic_series}
Assume that $\mathcal{L}(s)$ in Eq.~\eqref{eq:ME_rescaled} is $C^{k+2}$ ($k+2$ times differentiable),
TPHP for each fixed $s\ge0$, satisfies the gap condition and has
unique steady state. We denote by $\sigma(s)$ the unique steady
state of $\mathcal{L}(s)$, i.e., $\mathcal{L}(s)\sigma(s)=0$, $\Tr[\sigma(s)]=1$. Let $\rho(s)$ denote the solution of Eq.~\eqref{eq:ME_rescaled} with the initial
condition $\rho(0)=\sigma(0)$. Then
\bes
\begin{align}
\rho(s) & =\sigma(s)+\sum_{n=1}^{k}\zeta^{n}b_{n}(s)+\zeta^{k+1}r_{k}(\zeta,s)\label{eq:adia_series1}\\
b_{1}(s) & =S(s)\dot{P}(s)\sigma(s)=S(s)\dot{\sigma}(s)\label{eq:b1}\\
b_{n+1}(s) & =S(s)\dot{b}_{n}(s),\quad n=1,2,\ldots\label{eq:bn}
\end{align}
\ees
where 
$S(s)$ is the reduced resolvent, i.e., 
\beq
S(s)=\lim_{z\to0}Q(s)\left(\mathcal{L}(s)-z\right)^{-1}Q(s) ,
\eeq
and the remainder is 
\beq
\label{eq:remainder}
r_{k}(\zeta,s)=b_{k+1}(s)-\mathcal{E}(s,0)b_{k+1}(0)-\int_{0}^{s}\mathcal{E}(s,s')\dot{b}_{k+1}(s')ds'.
\eeq
\end{prop}
The proof is provided in Appendix~\ref{app:proof-prop-1}.

Note that the Liouvillian has dimension of 1/time. We could rescale $\mathcal{L}(s) = \frac{1}{\tau_0} \tilde{\mathcal{L}}(s)$ where $\tilde{\mathcal{L}}(s)$ is now dimensionless and $\tau_0$ is the natural time-scale of the process. The exact value of $\tau_0$ is to some extent arbitrary (it can be fixed by fixing the norm of $\tilde{\mathcal{L}}(s)$ at some $s$). More concretely, in quantum information processing experiments, $\mathcal{L}(s)$ is typically a perturbation of some Hamiltonian evolution, and so it is reasonable to take $\tau_0 = 1/J$ where $J$ is the energy scale of the Hamiltonian (we use units in which $\hbar =1$). After this rescaling all the formulas remain unchanged and  $\tau \mapsto \tau/\tau_0$. We see then that in Eq.~\eqref{eq:adia_series1} the expansion parameter is effectively the appropriately dimensionless quantity $\zeta = \tau_0/\tau$, while all the other quantities are also dimensionless. This expansion parameter is small when $\tau \gg \tau_0$ where $\tau$ is the timescale on which we change the Liouvillian. 

The following is a strengthening of a similar result contained in
\cite{avron_adiabatic_2012}, and introduces the assumption of vanishing boundary derivatives.
\begin{prop}
\label{prop:boundary_cancellation}
Under the same assumptions as in Proposition~\ref{prop:adiabatic_series}, with the additional assumptions that $\mathcal{L}(s)$ is independent of $\tau$, generates
a contraction semigroup, i.e., $\left\Vert e^{r\mathcal{L}(s)}\right\Vert _{1}\le1$
for each $r,s\ge0$, and that $\mathcal{L}^{(j)}(1)=0$ for $j=1,2,\ldots,k$ (vanishing derivatives at the final time):
\begin{equation}
\left\Vert \rho(1)-\sigma(1)\right\Vert_1 \le\frac{C_{k}}{\tau^{k+1}} ,
\label{eq:bound_BC}
\end{equation}
where $C_{k}$ is a constant independent of $\tau$.
\end{prop}

\begin{proof}
We first note that if $\mathcal{L}^{(j)}(s_{0})=0$ for $j=1,2,\ldots,k$
then, $\partial_{s}^{(j)}[(\mathcal{L}(s)-z)^{-1}]_{s=s_{0}}=0$.
This follows from 
\begin{equation}
\frac{\partial}{\partial s}\frac{1}{\mathcal{L}-z}=-\frac{1}{\mathcal{L}-z}\dot{\mathcal{L}}\frac{1}{\mathcal{L}-z}
\end{equation}
and iterating. Since the projector can be written as 
\begin{equation}
P(s)=\frac{1}{2\pi i}\oint_{\gamma}\frac{dz}{z-\mathcal{L}(s)}
\end{equation}
where $\gamma$ is a path that encircles only the zero eigenvalue
in anti-clockwise direction, it follows that also $P^{(j)}(s_{0})=0$
for $j=1,2,\ldots,k$. Moreover, it also follows immediately that $Q^{(j)}(s_{0})=S^{(j)}(s_{0})=0$
and $\sigma^{(j)}(s_{0})=P^{(j)}(s_{0})\sigma(s_{0})$ for $j=1,2,\ldots,k$.
We now use the assumptions and Proposition~\ref{prop:adiabatic_series}.
From Eqs.~\eqref{eq:b1} and \eqref{eq:bn} we see that $b_{n}(1)$ is
a sum of products of terms which contain $P$ and $S$ and their derivatives
up to order $n$ at $s=1$. All of these derivatives vanish up to order $k$, and so $b_{n}(1)=0$ $\forall n\leq k$. Hence 
\beq
\rho(1)=\sigma(1)+\zeta^{k+1}r_{k}(\zeta,1).
\eeq
At this point we need to bound the error $r_{k}(\zeta,1)$. Since by assumption $\mathcal{L}(s)$ generates a contraction, it follows that $\mathcal{E}(s_{1},s_{0})$ is a contraction
for $s_{1}\ge s_{0}$ (simply use the Trotter formula to write the propagator
as a limit of products), and we can bound the final remainder as: 
\bes
\begin{align}
\left\Vert r_{k}(\zeta,1)\right\Vert_{1} & \le\left\Vert b_{k+1}(s)\right\Vert _{1}+\left\Vert \mathcal{E}(s,0)b_{k+1}(0)\right\Vert_{1} \nonumber \\
 & +\int_{0}^{s}\left\Vert \mathcal{E}(s,s')\dot{b}_{k+1}(s')\right\Vert_{1} ds' 
 \label{eq:bound_r1}\\
 & \le\Big(\left\Vert b_{k+1}(1)\right\Vert_{1} +\left\Vert b_{k+1}(0)\right\Vert_{1} \nonumber \\
 & +\sup_{s\in[0,1]}\left\Vert \dot{b}_{k+1}(s)\right\Vert_{1} \Big)=:C_{k}.
 \label{eq:bound_r2}
\end{align}
\ees
The quantity in Eq.~\eqref{eq:bound_r2} does not depend on $\tau$
and is bounded because $S(s),P(s)$ and their derivatives are bounded
by the assumption that $\mathcal{L}(s)$ is smooth. 
\end{proof}
As noted above, there is an asymmetry between the boundary cancellation result for
dissipative generators (that can admit a one-dimensional kernel) and
the corresponding result for unitary evolutions. In the latter case,
in order to have the analogue of Eq.~\eqref{eq:bound_BC} one needs
the derivatives of the generators to be zero \emph{both at the end
and at the beginning} of the evolution (see \cite{garrido_degree_1962})
up to order $k$. The technical reason is that $P(s)$ must be rank
$1$ (ergodicity) and that the kernel of $P(s)$ must be independent of $s$. The latter
condition follows from trace preservation, i.e., conservation of probabilities,
of the evolution map $\mathcal{E}$. However the rank $1$ condition
cannot be satisfied by unitary dynamics. In other words, the difference
between Proposition~\ref{prop:boundary_cancellation} and the corresponding
result for the unitary case is due to the fact that in the former
we are dealing with irreversible dynamics, i.e., there is an arrow
of time. 

Two caveats should also be noted. First, while $C_{k}$ does not depend on $\tau$, it does contain an implicit dependence on the system size, and in general will grow with the latter, necessitating a corresponding growth of $\tau$ in order to keep the adiabatic error small.
Second, since in physical models, the generator $\mathcal{L}$ also depends 
on the bath, it may seem impossible to fulfill the condition $\mathcal{L}^{(j)}(1)=0$
for $j=1,2,\ldots,k$ by controlling only the system. As we show later, this
pessimistic view fortunately turns out to be wrong. 

Proposition~\ref{prop:boundary_cancellation} guarantees that as long as $\tau \gg \tau_0$ the adiabatic error can be made arbitrarily small. More precisely (switching to the rescaled generator), taking 
\begin{equation}
\tau \ge \tau_0 \left ( \frac{\tilde{C}_{k}}{\epsilon} \right )^\frac{1}{k+1}, 
\end{equation}
where $\tilde{C}_{k}$ refers now to the dimensionless generator $\tilde{\mathcal{L}}(s)$, implies 
\begin{equation}
\left\Vert \rho(1)-\sigma(1)\right\Vert _{1}\le \epsilon.
\end{equation}
If the constants $\tilde{C}_{k}$ were independent of $k$, this would imply a $(k+1)$-root speedup with respect to the case $k=0$. This hypothesis is likely overly optimistic: in the next subsection, using fairly crude bounds, we derive estimates for $\tilde{C}_{k}$ which predict a strong dependence on $k$. On the other hand, our numerical results in Sec.~\ref{sec:application}are encouraging especially for large $\tau$ and show that asymptotically $\left\Vert \rho(1)-\sigma(1)\right\Vert _{1} \sim \tau^{-(k+1)}$ (see Fig.~\ref{fig:numerics_Davies}).

\subsection{Stability and time-scales}

We next consider the practical feasibility
of the boundary cancellation approach by asking what happens if we try to set a derivative to zero
but only achieve a small norm? We focus on the
case where one tries to set the first derivative of the generator
to zero. We use Eq.~\eqref{eq:adia_series1} with $k=0$ and $k=1$
and bound the difference $\rho(1)-\sigma(1)$. In Appendix~\ref{sec:Explicit-Constants-for}
we show that then:
\begin{equation}
\left\Vert \rho(1)-\sigma(1)\right\Vert _{1}\le\min\left\{ \frac{B_{0}}{\tau},\frac{A_{1}}{\tau}+\frac{B_{1}}{\tau^{2}}\right\} ,
\label{eq:nonvanishing}
\end{equation}
where the constants are given by (all the superoperator norms are
induced, 1-1 norms)
\bes
\begin{align}
& B_{0}  =\left.\left(\left\Vert S\right\Vert ^{2}\left\Vert \dot{\mathcal{L}}\right\Vert \right)\right|_{0}^{1} +  \sup_{s\in[0,1]}\left(6\left\Vert S\right\Vert ^{3}\left\Vert \dot{\mathcal{L}}\right\Vert ^{2}+\left\Vert S\right\Vert ^{2}\left\Vert \ddot{\mathcal{L}}\right\Vert \right) \label{eq:B0}\\
& A_{1}  =\left\Vert S(1)\right\Vert ^{2}\left\Vert \dot{\mathcal{L}}(1)\right\Vert <B_{0} \label{eq:A1} \\
& B_{1}  =\left.\left(5\left\Vert S\right\Vert ^{4}\left\Vert \dot{\mathcal{L}}\right\Vert ^{2}+\left\Vert S\right\Vert ^{3}\left\Vert \ddot{\mathcal{L}}\right\Vert \right)\right|_{0}^{1} \label{eq:B1} + \\
 &\ \  \sup_{s\in[0,1]}\left(60\left\Vert S\right\Vert ^{5}\left\Vert \dot{\mathcal{L}}\right\Vert ^{3}+19\left\Vert S\right\Vert ^{4}\left\Vert \dot{\mathcal{L}}\right\Vert \left\Vert \ddot{\mathcal{L}}\right\Vert +\left\Vert S\right\Vert ^{3}\left\Vert \dddot{\mathcal{L}}\right\Vert \right)\notag
\end{align}
\label{eq:bounds}
\ees
and we used the notation $\left.\left(X\right)\right|_{0}^{1}=X(0)+X(1)$. 
Note that since  $A_{1} <B_{0}$ the minimum in Eq.~\eqref{eq:nonvanishing} is achieved by the second (first) term in the region where $\tau$ is large (small). 
We now imagine changing the schedule in order to try to enforce boundary cancellation but we only achieve it imperfectly. We obtain a new generator $\mathcal{L}'(s)$ with derivatives  of reduced norm for  $s$ close to the end (but not strictly zero)
and hence obtain new constants $A_1' <A_1, B_{0}',  B_{1}'$. In principle the modified schedule can unpredictably change $B_{0}$ and $B_{1}$. However, we have the following monotonicity result:
\begin{prop}
\label{prop:improvement} Assume that $\mathcal{L}(s),\mathcal{L}'(s)\in C^{3}([0,1])$
and we change the schedule only close to the end of the anneal, i.e.,
$\mathcal{L}'(s)=\mathcal{L}(s)$ for $s\in[0,1-\delta]$, and that
$\Vert\mathcal{L}'^{(j)}(s)\Vert<\Vert\mathcal{L}^{(j)}(s)\Vert$
for $s$ in a neighborhood of $s=1$ independent of $\delta$, for
$j=1,2,3$. Then for sufficiently small $\delta$, boundary cancellation
provides an improvement for all values of the anneal time $\tau$. 
\end{prop}
\begin{proof}
By assumption $A_{1}'<A_{1}$. Consider $B_{0}$  given by Eq.~\eqref{eq:B0}.
We have $\Vert S'\Vert^{2}\Vert\dot{\mathcal{L'}}\Vert(0)=\Vert S\Vert^{2}\Vert\dot{\mathcal{L}}\Vert(0)$, while $\Vert S'\Vert^{2}\Vert\dot{\mathcal{L'}}\Vert(1)<\Vert S\Vert^{2}\Vert\dot{\mathcal{L}}\Vert(1)$.
Now, consider the supremum term in Eq.~\eqref{eq:B0}, which we write
as $Y=\sup_{s\in[0,1]}X(s)=X(s_{0})$; after changing the schedule
we obtain $Y'=\sup_{s'\in[0,1]}X'(s')=X'(s'_{0})$. If $Y'\le Y$
this schedule is good enough and we keep it. Conversely, assume that 
$X'(s_{0}')>X(s_{0})$. By hypothesis $X'(s)=X(s)$ for $s\in[0,1-\delta)$. We can now take $\delta$ small enough such that
$[1-\delta,1]$ is entirely in the region where $\Vert\mathcal{L}'^{(j)}(s)\Vert<\Vert\mathcal{L}^{(j)}(s)\Vert$
(for $j=1,2$). At this point we must have necessarily $X'(s_{0}')\le X(s_{0})$.
Hence $B_{0}'<B_{0}$ by Eq.~\eqref{eq:B0}. An entirely analogous argument holds for $B_{1}$,
and so there exist a $\delta$ small enough such that $A_{1}'<A_{1},B_{0}'<B_{0}$, and $B_{1}'<B_{1}$.
This implies that the adiabatic error, as predicted by Eq.~\eqref{eq:nonvanishing}, is smaller after
boundary cancellation is employed. 
\end{proof}

\begin{figure}
\begin{centering}
\includegraphics[width=6.1cm]{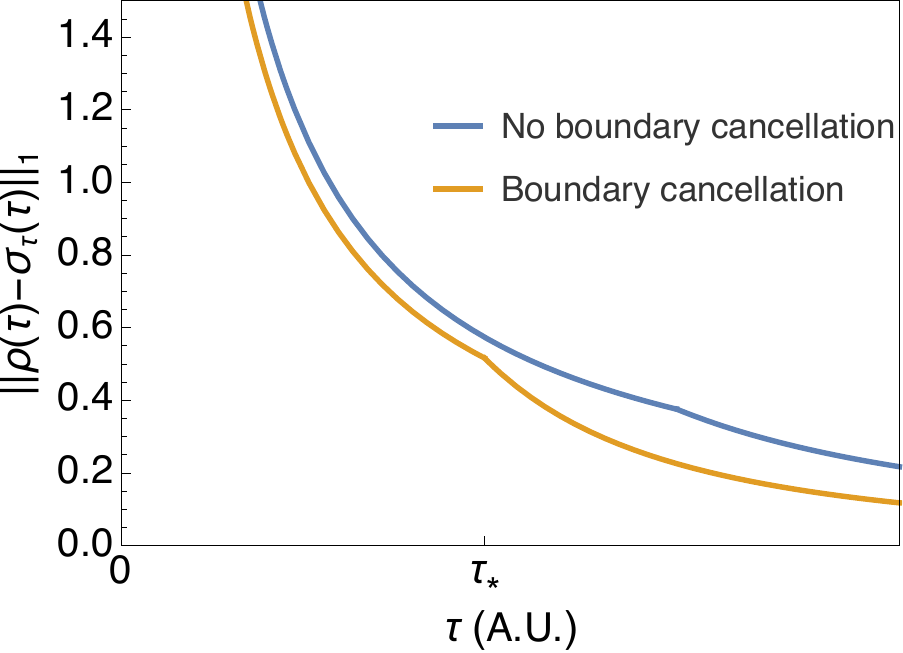}
\par\end{centering}
\caption{Schematic plot of the adiabatic error before and after imperfect boundary cancellation at first order. The blue line shows Eq.~\eqref{eq:nonvanishing} for random values of $A_1,B_0,B_1$, while for the orange line these values are (slightly) diminished in accordance with the conditions of Proposition~\ref{prop:improvement}.  There is an improvement for all values of $\tau$, although a larger improvement is predicted for $\tau>\tau_\ast = B'_{1}/\left(B_{0}'-A_{1}'\right)$.
\label{fig:bound_improved}}
\end{figure}

See Fig.~\ref{fig:bound_improved} for a plot of the improvement predicted by boundary cancellation under these circumstances. Note that our discussion is framed in terms of upper bounds, and it is possible that a larger benefit exists in a larger region than $\tau >B'_{1}/\left(B_{0}'-A_{1}'\right)$. 

Estimating the scaling of the terms in Eq.~\eqref{eq:bounds} as a function of relevant parameters such as size and temperature, is of great importance for applications, e.g., the preparation of
ground states or thermal Gibbs states in quantum annealing, where we are interested in the behavior of the adiabatic
error $\left\Vert \rho_{\tau}(\tau)-\sigma_{\tau}(\tau)\right\Vert _{1}$
with respect to the system size $N$ (number of qubits). One is then led to estimate the norm of the (reduced) resolvent.
In contrast to the closed-system case the norm of $S$ cannot be
simply evaluated, i.e.: 
\begin{equation}
\left\Vert S(s)\right\Vert _{1,1}\neq\frac{1}{\mathrm{dist}(0,\sigma(\mathcal{L}(s)))}
\end{equation}
where $\sigma(\mathcal{L}(s))$ is the spectrum of $\mathcal{L}(s)$.
This fact makes the estimates complicated.  For $k=0$ in the low temperature regime, 
such estimates were given in \cite{venuti_adiabaticity_2016,venuti_relaxation_2017}. The result is that $\left\Vert S(s)\right\Vert _{1,1}$ depends not only on the Liouvillian gap but also on the Hamiltonian one. 

\section{Application of boundary cancellation while controlling only the system Hamiltonian}
\label{sec:application}

Our goal is to apply the boundary cancellation method under realistic conditions using
master equations for time-dependent system-Hamiltonians.
I.e., given a total Hamiltonian $H_{\mathrm{tot}}(t)=H_{S}(t)+H_{I}+H_{B}$,
the sum of system, interaction, and bath Hamiltonians respectively, we wish to consider master equations in the form of Eq.~\eqref{eq:ME} with $\mathcal{L}_{\tau}(t)$ derived from first principles, while directly controlling the boundary terms of only the system Hamiltonian. We will consider three such master equations. Henceforth we write the interaction Hamiltonian explicitly in the general form 
$H_{I}=g\sum_{\alpha}A_{\alpha}\otimes B_{\alpha}$.

\subsection{Master equations}
The first master equation is 
the Davies-Lindblad adiabatic master equation (DLAME) derived in \cite{albash_quantum_2012}.
%
Its generator is
given by  
\bes
\label{eq:davies}
\begin{align}
\mathcal{L}_{\tau}(t) & =-i\left[H_{S}(t)+H_{LS}(t),\bullet\right] + \mathcal{L}^D_{\tau}(t)  \\
 \mathcal{L}^D_{\tau}(t) & =\sum_{\alpha,\beta,\omega_{n}}\gamma_{\alpha,\beta}(\omega_{n})\Big(A_{\beta}(\omega_{n})\bullet A^{\dagger}_{\alpha}(\omega_{n})\nonumber \\
 & -\frac{1}{2}\left\{ A^{\dagger}_{\alpha}(\omega_{n})A_{\beta}(\omega_{n}),\bullet\right\} \Big).\end{align}
\ees
Here $H_{LS}(s)$ is the Lamb-shift term, $\gamma_{\alpha,\beta}(\omega)$
is the Fourier transform of the bath-correlation function
\beq
G_{\alpha,\beta}(t,s):=g^{2}\langle B_{\alpha}(t)B_{\beta}(s)\rangle=G_{\alpha,\beta}(t-s) ,
\eeq
and $\omega_{n}(s)$ are the Bohr frequencies of $H_S(s)$ (to simplify notation we suppress their explicit time-dependence when convenient). The Lindblad jump operators $A_{\alpha}(\omega_{n})$
that appear in the Davies generator $\mathcal{L}^D$ are given by
\begin{equation}
e^{itH(s)}A_{\alpha}e^{-itH(s)}=\sum_{\omega_{n}}e^{-it\omega_{n}}A_{\alpha}(\omega_{n}),
\label{eq:jump}
\end{equation}
and
\beq
H_{LS} = \sum_{\alpha,\beta,\omega_{n}} S_{\alpha,\beta}(\omega_n)A^\dagger_{\alpha}(\omega_n)A_{\beta}(\omega_n) ,
\eeq
with
\beq
S_{\alpha,\beta}(\omega) = \int_{-\infty}^{\infty} d\omega' \gamma_{\alpha\beta}(\omega') \mathcal{P}\left(\frac{1}{\omega-\omega'}\right) ,
\eeq
where $\mathcal{P}$ is the Cauchy principal value.
This master equation preserves complete positivity and assumes that the system evolves adiabatically.

The second master equation is the the Schr\"{o}dinger picture Redfield master equation (SPRME), which we write as:
\begin{align}
&\frac{\partial\rho(t)}{\partial t}  =-i\left[H_{S}(t),\rho(t)\right]\nonumber \\
 & +\sum_{\alpha,\beta}\int_{0}^{t}drG_{\alpha,\beta}(r)\,\left[U_{0}(t,t-r)A_{\beta}U_{0}(t-r,t)\rho(t),A_{\alpha}\right]\nonumber \\
 & +\hc  ,
 \label{eq:nonadia_red}
\end{align}
where the unperturbed propagator is generated purely by the system Hamiltonian, i.e., is the solution of $\partial_t{U}_{0}(t,0)=-iH_{S}(t)U_{0}(t,0)$
with the boundary condition $U_{0}(0,0)=\1$.
The SPRME is a generalization of the standard Redfield master equation, which is typically written in the interaction picture for time-independent system Hamiltonians~\cite{breuer_theory_2007}.

The third master equation is obtained after performing an adiabatic-type approximation on Eq.~\eqref{eq:nonadia_red}, so we call it the adiabatic Redfield master
equation (ARME):
\begin{align}
 \label{eq:adia_red}
\frac{\partial\rho(t)}{\partial t}  & =-i\left[H_{S}(t),\rho(t)\right] \\
 & +\sum_{\alpha,\beta}\int_{0}^{\infty}drG_{\alpha,\beta}(r)\,\left[A_{\beta}(-r,t)\rho(t),A_{\alpha}\right]+\hc \nonumber
\end{align}
where 
\beq
A_{\beta}(-r,t)=e^{-irH_{S}(t)}A_{\beta}e^{irH_{S}(t)}.
\eeq

We derive the SPRME and the ARME in Appendix~\ref{app:Redfield}, where we also estimate the error of the approximations involved. As far as we know these forms of the Redfield equation have not appeared previously. The SPRME has the advantage that it tolerates a bath with algebraically decaying correlation functions, the limiting case being $G_{\alpha,\beta} (t) \sim t^{-2}$, which results in a relative error growing as $\ln(\tau)$. For the same bath the ARME introduces an error $\propto \tau$; more details are given in Appendix~\ref{app:Redfield}.
We note that the Davies generator is obtained from the ARME after the rotating wave (secular) approximation, i.e., the Redfield case requires one fewer approximations. However, while Redfield theory is TPHP, unlike the Davies-Lindblad case it is notoriously not completely positive (though various fixes have been proposed~\cite{Gaspard:1999aa,Whitney:2008aa}). 

\subsection{Application of boundary cancellation}

We now investigate whether it is possible to satisfy the assumptions of Proposition \ref{prop:boundary_cancellation} under realistic conditions.

\subsubsection{The Davies-Lindblad adiabatic master equation}
We begin by considering the DLAME, Eq.~\eqref{eq:davies}. 
We assume henceforth that the system Hamiltonian is a function of
$t/\tau$ and not separately of $t$ and $\tau$. Note that for the time-dependent Davies generator this implies
that the rescaled generator $\mathcal{L}(s)$ resulting from $\mathcal{L}_{\tau}(t)$ is $\tau$-independent.
\begin{prop}
\label{prop:adia_davies}
Assume a master equation with generator in
time dependent Davies form $\mathcal{L}_{\tau}(t)$ given by Eq.~\eqref{eq:davies}.
Moreover assume that the system Hamiltonian $H_{S}(t)$ is smooth
and that the degeneracy of all the levels does not change for $t\in[0,\tau]$.
If $\partial_{t}^{(j)}[H_{S}(t)]_{t=\tau}=0$ for $j=1,2,\ldots,k$ then $\partial_{t}^{(j)}[\mathcal{L}_{\tau}(t)]_{t=\tau}=0$
for $j=1,2,\ldots,k$. Furthermore, if the steady state of $\mathcal{L}_{\tau}(t)$
is unique for $t\in[0,\tau]$ then the assumptions of Proposition
\ref{prop:boundary_cancellation} all hold, so that
\begin{equation}
\left\Vert \rho_{\tau}(\tau)-\sigma_{\tau}(\tau)\right\Vert _{1}\le\frac{C_{k}}{\tau^{k+1}}.\label{eq:BC}
\end{equation}
\end{prop}

\begin{proof}
The degeneracy assumption is needed since otherwise $\mathcal{L}_{\tau}(t)$
is not even continuous. Let $H_{S}(t)=\sum_{n}E_{n}(t)\Pi_{n}(t)$, where 
$E_{n}(t)$ and $\Pi_{n}(t)$ are the instantaneous energies and eigenprojectors, respectively.
The assumptions imply that $E_{n}^{(j)}(\tau)=0$ for $j=1,2,\ldots,k$
and so the same holds for the Bohr frequencies $\omega_{n}^{(j)}(\tau)$. The Lindblad jump
operators $A_{\alpha}(\omega_{n})$ that appear in the Davies generator
appear in Eq.~\eqref{eq:jump}, whereby:
\begin{equation}
A_{\alpha}(\omega_{n})=\lim_{X\to\infty}\frac{1}{X}\int_{0}^{X}dt'e^{it'\omega_{n}(t)}e^{it'H(t)}A_{\alpha}e^{-it'H(t)}.
\end{equation}
Now, using the Duhamel formula 
\begin{equation}
\partial_{t}e^{B(t)}=\int_{0}^{1}dre^{rB(t)}\left(\partial_{t}B\right)e^{(1-r)B(t)}
\end{equation}
repeatedly, together with $\omega_{n}^{(j)}(\tau)=0$ for $j=1,2,\ldots,k$,
one obtains 
\begin{equation}
\partial_{t}^{(j)}\left[A_{\alpha}(\omega)\right]_{t=\tau}=0,\,\,\mathrm{for}\,\,\,j=1,2,\ldots,k ,
\end{equation}
which ensures that both $\partial_{t}^{(j)}[\mathcal{L}^D_{\tau}]_{t=\tau}$ and $\partial_{t}^{(j)}[H_{LS}]_{t=\tau}$ [Eq.~\eqref{eq:davies}] vanish, so that the Proposition~\ref{prop:boundary_cancellation} assumption that 
$\mathcal{L}^{(j)}(1)=0$ for $j=1,2,\ldots,k$ is also satisfied.

To prove the last assertion of the current Proposition note that the assumptions,
together with finite dimensionality, imply that the zero eigenvalue
is separated by a finite gap from the rest of the spectrum. Moreover,
a theorem due to Kossakowski \cite{kossakowski_quantum_1972} (see
also Theorem 3.3.1 of \cite{rivas_open_2012}) states that a Lindbladian generates
a contraction semigroup, so we can apply Proposition \ref{prop:boundary_cancellation}.
More specifically, the Lindbladian assumption implies that $\left\Vert \mathcal{E}(s,s')\right\Vert _{1}\le1$
for $s\ge s'$, so that we can go from Eq.~\eqref{eq:bound_r1} to
Eq.~\eqref{eq:bound_r2}. Since $\mathcal{L}(s)$ is independent of $\tau$,
the right hand side of Eq.~\eqref{eq:bound_r2} is (bounded and) independent
of $\tau$ and the result follows. 
\end{proof}

\begin{figure}[t]
\subfigure[\ $T=1$mK]{\includegraphics[width=7cm]{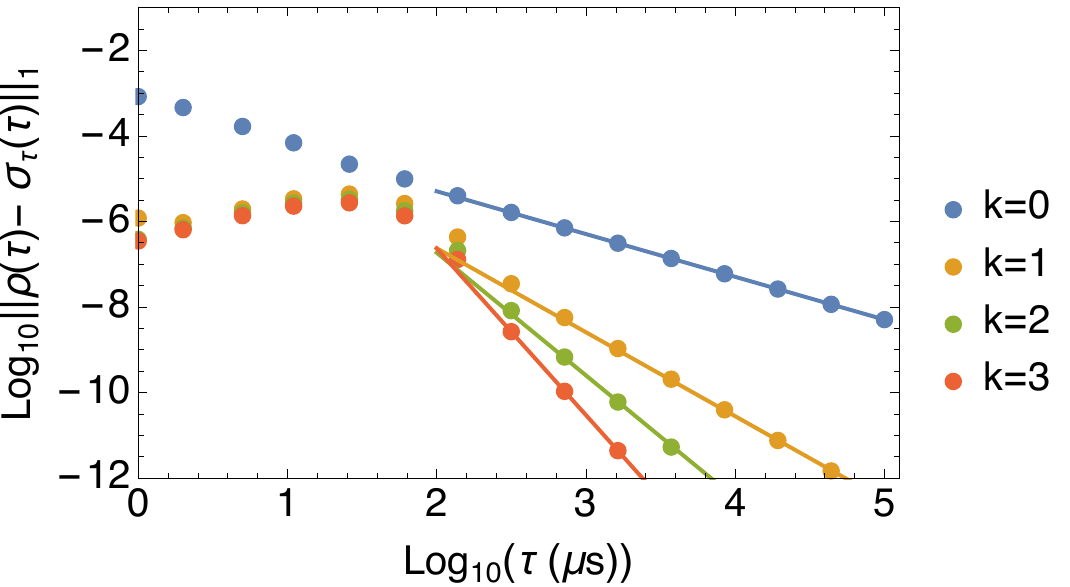}}
\subfigure[\ $T=12$mK]{\includegraphics[width=7cm]{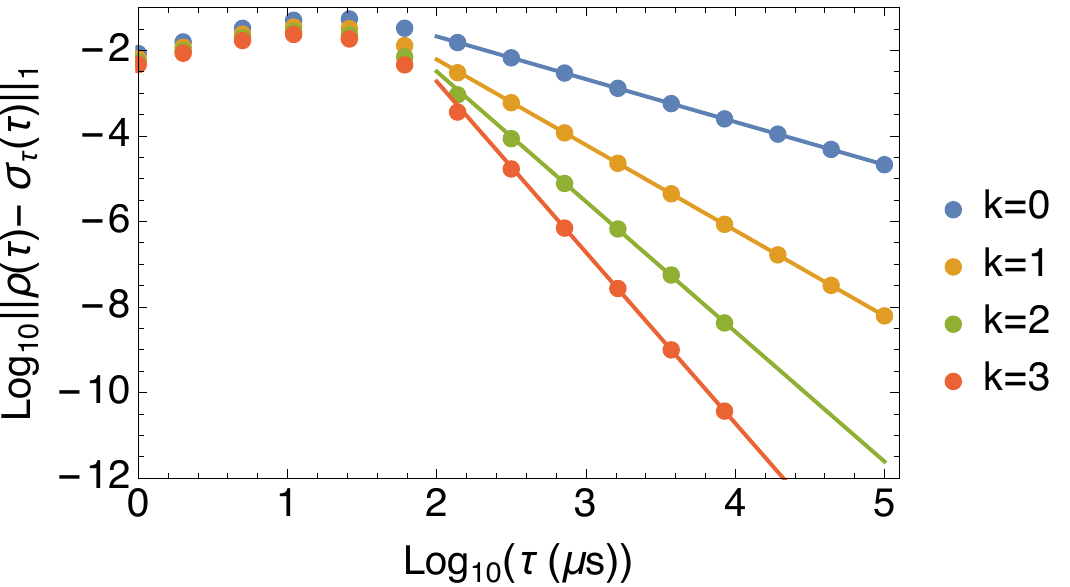}}
\subfigure[\ $T=20$mK]{\includegraphics[width=7cm]{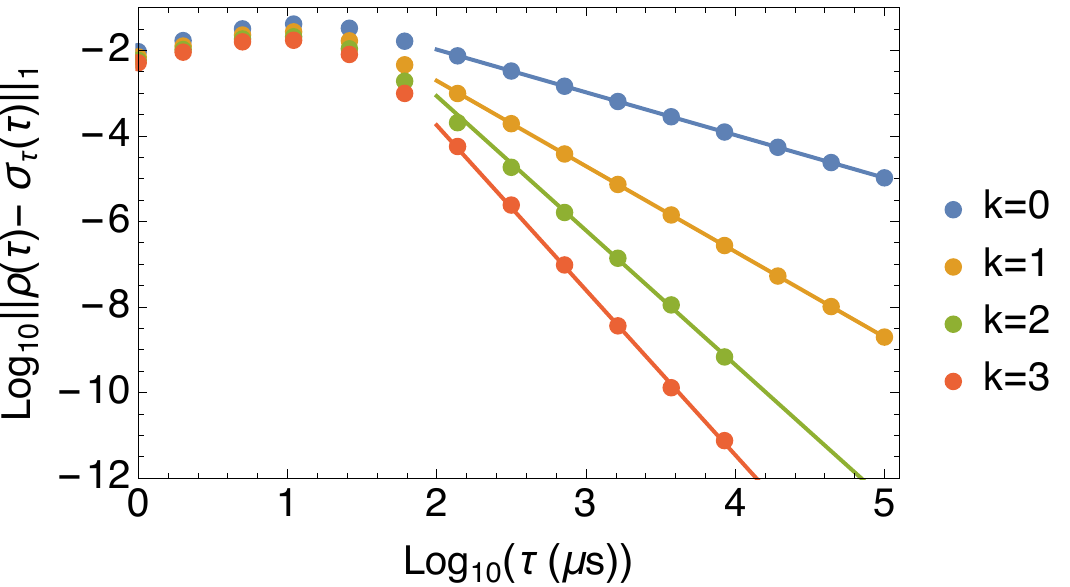}}
\caption{Adiabatic error $\left\Vert \rho_{\tau}(\tau)-\sigma_{\tau}(\tau)\right\Vert _{1}$
 as a function of annealing time $\tau$ using the boundary cancellation
method for different $k$'s and different temperatures, for the Davies-Lindblad adiabatic master equation with an Ohmic bath.
Parameters are: $g=10^{-5/2}$GHz $=3.16$MHz, $\omega_{c}=8\pi$GHZ$=25.13$GHz and $\eta =1$(ns)$^2$. These parameters can describe experiments with flux qubits \cite{Harris:2010kx,Weber:2017aa}. 
The continuous lines are best fits of the form $\sim \tau^{-\alpha_k}$. The resulting exponents  for $k=(0,1,2,3)$ are given by (a) $\alpha=(1.00,1.96,2.86,3.89)$, (b)
$\alpha=(0.99,1.99,3.03,3.99)$, (c) $\alpha=(0.99,1.99,3.14,3.85)$.
\label{fig:numerics_Davies}}
\end{figure}

Results of numerical simulations for a single qubit evolving according
to Eq.~\eqref{eq:davies} are shown in Fig.~\ref{fig:numerics_Davies}.
The Hamiltonian is taken to be $H_{S}(t)=\omega_{x}\sigma^{x}[1-\vartheta_{k}(t/\tau)]+\omega_{z}\sigma^{z}\vartheta_{k}(t/\tau)$.
The schedule is given by $\vartheta_{k}(s)=2B_{(s+1)/2}(k+1,k+1)/B_{1}(k+1,k+1)$
where $B_{s}(a,b)$ is the incomplete Beta function \cite{RPL:10}.
It has the property of having vanishing derivatives at $s=1$ up to
order $k$ but not for $k+1$. The system-bath operator is $A=\sigma^{y}$,
and the bath correlation function is Ohmic, i.e.,
\begin{equation}
\hat{G}(\omega):=\int_{-\infty}^{+\infty}e^{i\omega t}G(t)dt=g^{2} \eta 2\pi\frac{\omega e^{-\left|\omega\right|/\omega_{c}}}{1-e^{-\beta\omega}},
\end{equation}
where $\eta$ is a constant with dimension of time squared. The simulations are carried out for an annealing time $\tau$ that is sufficiently large for the asymptotic region to be reached, where $\left\Vert \rho_{\tau}(\tau)-\sigma_{\tau}(\tau)\right\Vert _{1}\sim\tau^{-\alpha_k}$, where $\alpha_k \approx k+1$ (fits in Fig.~\ref{fig:numerics_Davies}). Note that for the
DLAME, the instantaneous steady state is
given by the thermal Gibbs state: $\sigma_{\tau}(t)=\rho_{G}(t):=\exp\left(-\beta H_{S}(t)\right)/Z$. 

\subsubsection{The Redfield master equations}

We now turn our attention to the two types of Redfield master equations.
We first consider the ARME. 
\begin{prop}
\label{prop:adia_red}
Assume the adiabatic Redfield master equation [Eq.~\eqref{eq:adia_red}] holds. Moreover, assume that the system Hamiltonian $H_{S}(t)$ is smooth and 
$\partial_{t}^{(j)}[H_{S}(t)]_{t=\tau}=0$ for $j=1,2,\ldots,k$. Then $\partial_{t}^{(j)}[\mathcal{L}_{\tau}(t)]_{t=\tau}=0$ for $j=1,2,\ldots,k$.
As a consequence, the adiabatic expansion \eqref{eq:adia_series1} holds with $b_{n}(s)=0$
for $n=1,2,\ldots,k$. 
\end{prop}

\begin{proof}
The result follows simply by repeatedly taking the derivative of $A_{\beta}(-r,t)$
with respect to $t$ at $t=\tau$ using the Duhamel formula, exactly as in the proof of Proposition~\ref{prop:adia_davies}.
\end{proof}

In contrast to the DLAME case [Eq.~\eqref{eq:davies}], we cannot prove that the bound \eqref{eq:bound_BC} generally holds in the present case. The reason is that this requires bounding the error $r_{k}(\zeta,1)$ by a constant independent of $\tau$. However, since the ARME does not always generate a contraction, this is not always possible for all initial states and parameter values. 
We do not report numerical simulations for the ARME case since for the parameters chosen here the evolutions it generates turn out to be completely-positive and trace-preserving (CPTP), and hence the results will agree with Proposition~\ref{prop:adia_davies}.

We next consider the SPRME [Eq.~\eqref{eq:nonadia_red}].
In this case even Proposition~\ref{prop:adia_red} does not apply. 
For example, even if $\partial_{t}H_{S}(\tau)=0$, if we differentiate
the non-adiabatic generator once at $t=\tau$ with $\partial_{t}H_{s}(\tau)=0$
we obtain a term proportional to $G_{\alpha,\beta}(\tau)$ and a term
proportional to $\int_{0}^{\tau}r^{2}G_{\alpha,\beta}(\tau)dr=O\left(\tau_{B}^{3}g^{2}\right)$
under the assumption of a fast bath. However, both these terms are supposed
to be small and so one may hope that the conclusions of Proposition~\ref{prop:adia_davies} are qualitatively valid at least in some range
of parameters. 

\begin{figure}
\begin{centering}
\includegraphics[width=6.1cm]{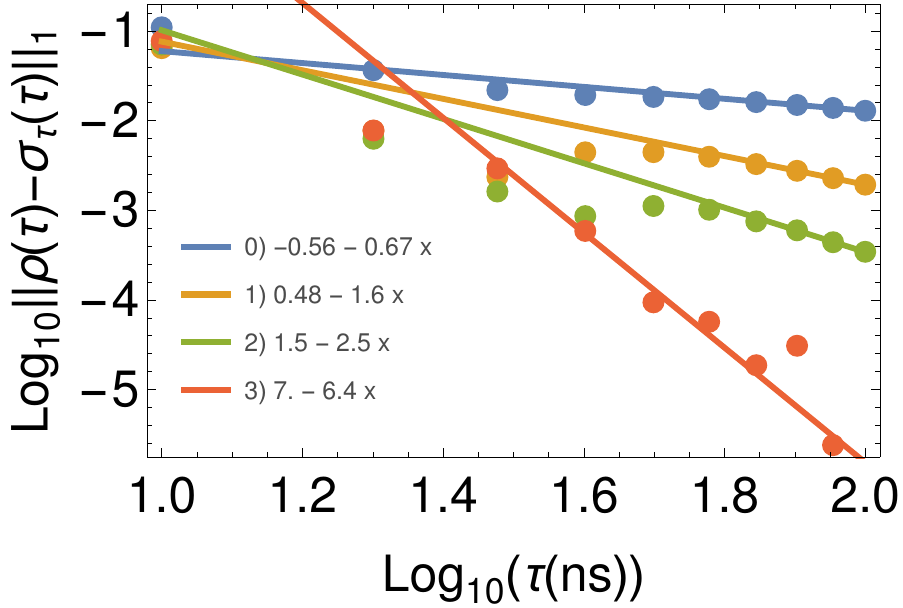}
\par\end{centering}
\caption{Adiabatic error $\left\Vert \rho_{\tau}(\tau)-\sigma_{\tau}(\tau)\right\Vert _{1}$
as a function of annealing time $\tau$ using the boundary cancellation method
for $k=0,1,2,3$ for the Schr\"odinger picture Redfield master equation \eqref{eq:nonadia_red}. 
The parameters are: $g=0.1$GHz , $\omega_{c}=16$GHz, $T=12$mK.
For $k=0,1,2$ the fit is obtained using the last four most significant
points, for $k=3$ the penultimate point has been excluded from the
fit. Note that the total annealing times here are much shorter than in
Fig.~\ref{fig:numerics_Davies} and the asymptotic
region where $\left\Vert \rho_{\tau}(\tau)-\sigma_{\tau}(\tau)\right\Vert _{1}\sim\tau^{-(k+1)}$ has not yet been reached. Boundary cancellation is seen to provide a consistent advantage for $\tau \gtrsim 20$ns.
\label{fig:nonadia_red}}
\end{figure}

In order to check the latter conjecture we performed numerical simulations
using Eq.~\eqref{eq:nonadia_red}. The results are shown in Fig.~\ref{fig:nonadia_red}, where we see that boundary cancellation improves the adiabatic error for sufficiently large annealing times. The slope seems to roughly track the $k+1$ rule expected if the assumptions of Proposition~\ref{prop:adia_davies} were to hold, but we caution that the asymptotic regime was not reached due to the heavy computational cost of these simulations. In fact, the numerical computation of the integral appearing in Eq.~\eqref{eq:nonadia_red} constitutes its own challenge; more details are given in Appendix~\ref{app:X}. 

\begin{figure}[t]
\includegraphics[width=7cm,keepaspectratio]{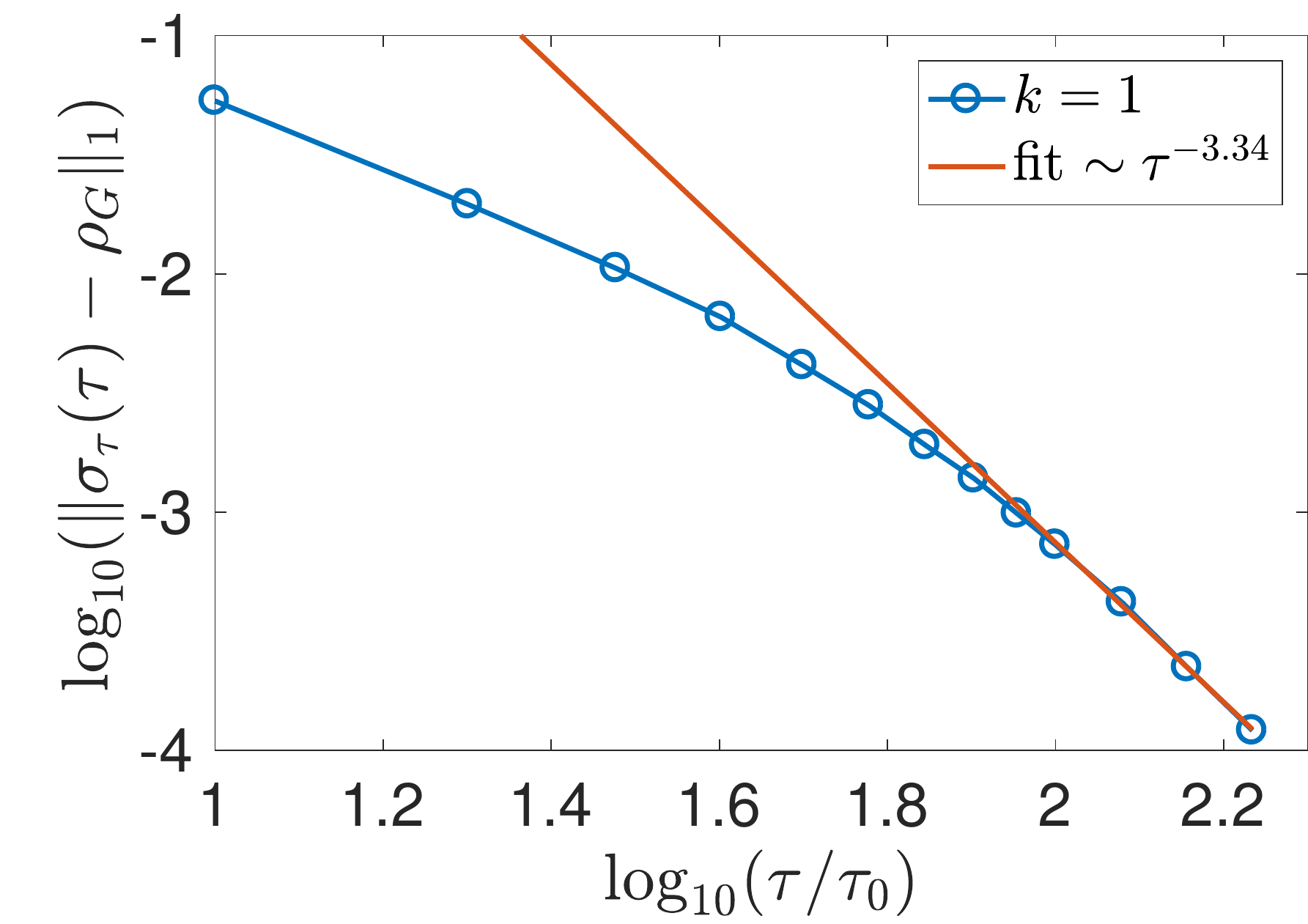}
\caption{Distance of the instantaneous steady state of the generator
\eqref{eq:nonadia_red} from the corresponding Gibbs state $\rho_{G}(\tau)$.
The timescale is $\tau_{0}=1$ns. Other parameters are the same as
in Fig.~\ref{fig:nonadia_red}. The same plot is obtained for different
values of $k$.}
\label{fig:redfield_details}
\end{figure}

Note that the generator in Eq.~\eqref{eq:nonadia_red} reduces to $-i\left[H_{S}(0),\bullet\right]$
at $t=0$. As such the instantaneous steady state is degenerate at
$t=0$. We fixed the initial state by taking $\sigma_{\tau}(0):=\lim_{t\to0^{+}}\sigma_{\tau}(t)$. It turns out that $\sigma_{\tau}(0)=\1/2$ in all of our simulations. At the other boundary $t=\tau$, the steady state approaches the thermal state $\rho_{G}(\tau)$.
In fact we have numerically checked that $\left\Vert \sigma_{\tau}(\tau)-e^{-\beta H_{S}(\tau)}/Z\right\Vert _{1}\sim\tau^{-\alpha}$; 
see Fig.~\ref{fig:redfield_details}. 

Let us also comment on complete positivity.
The precise, general
characterization of the region of parameters that ensure this condition is beyond the scope of this work. However, as shown in Fig.~\ref{fig:redfield_details2} for an Ohmic bath, complete positivity is violated at very low temperatures and large values of $\omega_{c}$.
This (counterintuitive) fact seems to be due to the presence of
fast oscillations appearing for large $\omega_{c}$.

\begin{figure}[t]
\ \ \ \
\includegraphics[width=7cm,keepaspectratio]{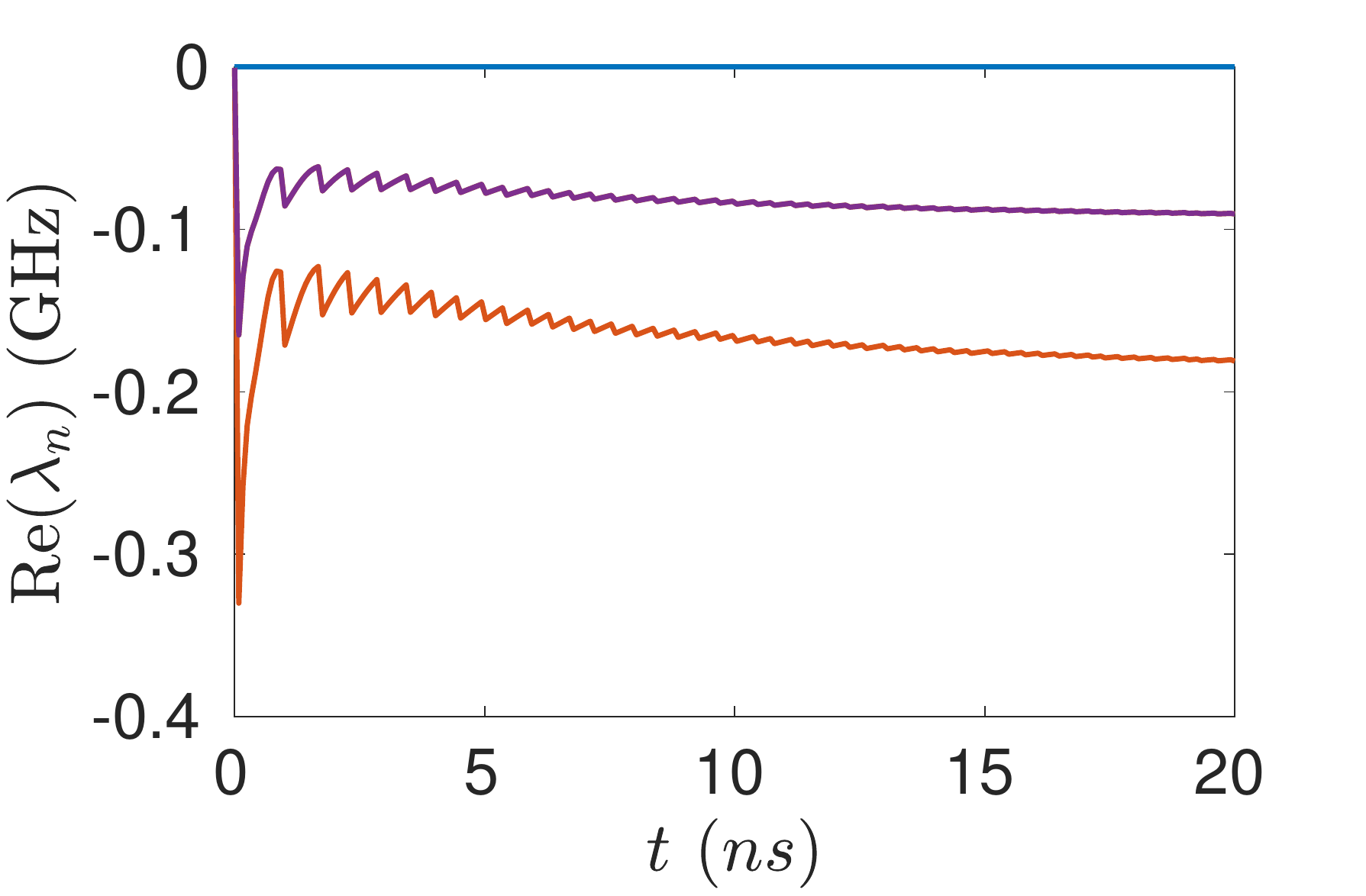}
\includegraphics[width=6.35cm,keepaspectratio]{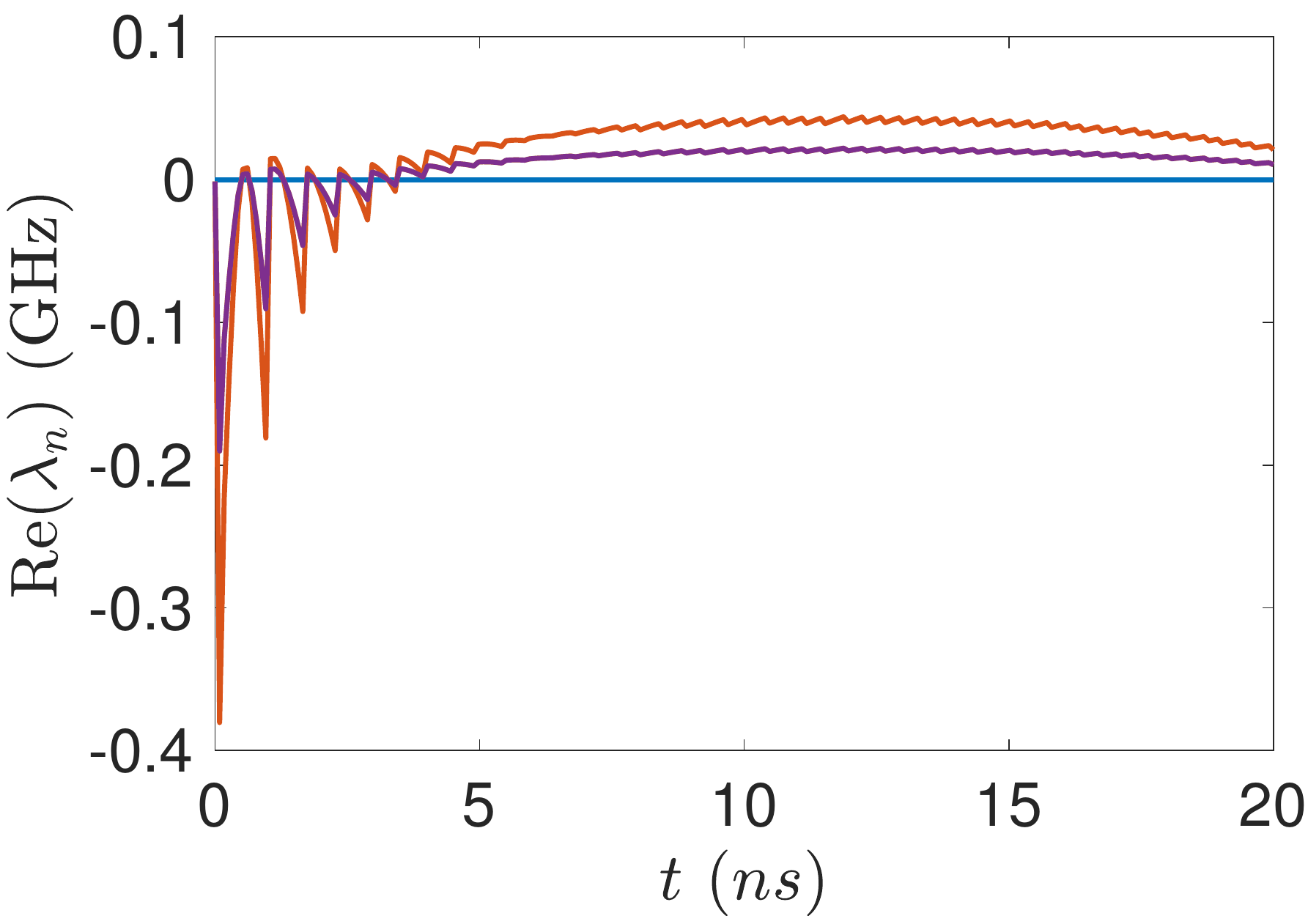}
\caption{Real part of instantaneous eigenvalues
of the generator \eqref{eq:nonadia_red}, for $\tau=20$ns and $k=0$. Top panel: 
$g=0.1$GHz , $\omega_{c}=16$GHz, $T=12$mK (as in Fig.~\ref{fig:nonadia_red}).
Bottom panel: $g=0.1$GHz, $\omega_{c}=25.13$GHZ,
$T=1$mK. In the bottom panel the instantaneous eigenvalues of the generator $\mathcal{L}_{\tau}(t)$
acquire a positive real part in some range $t\in[0,\tau]$, which means that the propagator
is not completely positive.}
\label{fig:redfield_details2}
\end{figure}

\subsubsection{The Hamiltonian case}
The next result (analogous to Theorem 4 of \cite{lidar_adiabatic_2009}) shows that the boundary cancellation result is 
stable with respect to non-Markovianity, and at the same time
that the lack of boundedness that can in principle emerge from the
Redfield master equation is unphysical. This requires that we assume vanishing derivatives at both the initial and final times.
\begin{prop}
Assume that the total Hamiltonian has the form
$H_{\mathrm{tot}}(s)=H_{S}(s)+H_{B}+H_{I}$, where only $H_{S}$
depends on the rescaled time $s$. Let $|\phi(s)\rangle$ denote the instantaneous
eigenstate of $H_{\mathrm{tot}}(s)$ related to some total energy level
and $|\psi(s)\rangle$ the Schr\"{o}dinger-evolved state starting
from $|\phi(0)\rangle$. Also assume that this level is separated
by a finite gap from the rest of the spectrum for all $s\in[0,1]$.
Let $\rho(s)=\Tr_{B}|\psi(s)\rangle\langle\psi(s)|$ and $\sigma(s)=\Tr_{B}|\phi(s)\rangle\langle\phi(s)|$.
If $H_{S}^{(j)}(0)=H_{S}^{(j)}(1)$ for $j=1,2,\ldots,k$ then
\beq
\left\Vert \rho(1)-\sigma(1)\right\Vert _{1}\le\frac{C_{k}}{\tau^{k+1}} .
\eeq
\end{prop}

\begin{proof}
Obviously, if $H_{S}^{(j)}(1)=0$ then also $H^{(j)}(1)=0$. We now apply the
analogous result of Proposition \ref{prop:boundary_cancellation}
for unitary dynamics, which requires the derivatives of the generator (the Hamiltonian)
to vanish also at $s=0$ \cite{garrido_degree_1962}. So
we have $\left\Vert |\psi(s)\rangle\langle\psi(s)|-|\phi(s)\rangle\langle\phi(s)|\right\Vert _{1}\le C_{k}\tau^{-(k+1)}$.
Since CPTP maps are contractions for the trace norm distance (i.e.,
 $\left\Vert \mathcal{E}\rho_{1}-\mathcal{E}\rho_{2}\right\Vert _{1}\le\left\Vert \rho_{1}-\rho_{2}\right\Vert _{1}$),
the result follows from the fact that the partial trace is a CPTP map:
\begin{align}
\left\Vert \rho(s)-\sigma(s)\right\Vert _{1}&\le\left\Vert |\psi(s)\rangle\langle\psi(s)|-|\phi(s)\rangle\langle\phi(s)|\right\Vert _{1} \notag \\ 
&\le\frac{C_{k}}{\tau^{k+1}} .
\end{align}
\end{proof}

\section{Summary and Conclusions}
\label{sec:conc}

We have generalized the boundary cancellation method to open systems
described by a time-dependent Liouvillian $\mathcal{L}_{\tau}(t)$.
If $\mathcal{L}_{\tau}(t)$ is ergodic (i.e., its instantaneous steady state is unique), generates a completely positive quantum map, and has vanishing
derivatives up to order $k$ at the end of the evolution, then the
adiabatic steady-state preparation error is upper bounded by $C_{k}/\tau^{k+1}$. Next,
we performed a detailed analysis to investigate whether the assumptions underlying this result can be satisfied in a realistic setting, where one controls only the system Hamiltonian. For the time-dependent Davies-Lindblad adiabatic master equation derived in \cite{albash_quantum_2012} 
the boundary cancellation result can indeed by achieved by requiring that the system Hamiltonian
has vanishing derivatives up to order $k$ only at the end of the evolution. To go beyond this setting we derived two time-dependent master equations of the Redfield type, one in the Schr\"odinger picture with a general time-dependent system Hamiltonian and the other under an additional adiabatic approximation. In this case the lack of complete positivity prevents the rigorous applicability of our result. However our numerical
simulations shows that boundary cancellation still 
holds for the Schr\"odinger picture Redfield master equation, for the range
of parameters where the evolution is positive. 
We have also shown analytically that the boundary cancellation result exhibits a degree of robustness in the sense that even if the derivatives do not exactly vanish but instead are upper bounded, then the adiabatic error can still be reduced, with the reduction being more pronounced  for longer annealing times. 

Boundary cancellation is a relatively straightforward method that
can be applied to experimental quantum annealers used to prepare steady states such as thermal Gibbs states. It allows for a smaller error at a given preparation time, or equivalently a shorter
preparation time at a given error, and should hence be used when possible.

\begin{acknowledgments}
The authors would like to thank Tameem Albash for producing the numerical simulations and fits shown in Fig.~\ref{fig:numerics_Davies}.
This research is based upon work partially supported by the Office of the Director of National Intelligence (ODNI), Intelligence Advanced Research Projects Activity (IARPA), via the U.S. Army Research Office contract W911NF-17-C-0050. The views and conclusions contained herein are those of the authors and should not be interpreted as necessarily representing the official policies or endorsements, either expressed or implied, of the ODNI, IARPA, or the U.S. Government. The U.S. Government is authorized to reproduce and distribute reprints for Governmental purposes notwithstanding any copyright annotation thereon.
\end{acknowledgments}

\appendix

\section{Proof of Proposition~\ref{prop:adiabatic_series}}
\label{app:proof-prop-1}
Proposition~\ref{prop:adiabatic_series} is a special case of Theorem 6 of \cite{avron_adiabatic_2012} and all we need to prove is that the terms $a_n(s)$ of Eq.~(15) in that reference satisfy $a_{0}(s)=T(s,0)a_{0}(0)=\sigma(s)$ and $a_{n}(s)=0$ for $n\ge1$, where $T(s,s')$ denotes the parallel transport (perfect adiabatic evolution) operator, which satisfies $P(s)T(s,s') = T(s,s')P(s')$.

\begin{proof}
We use the adiabatic series in \cite{avron_adiabatic_2012} and follow
the same notation therein with the only modification being that we replace $\mathcal{L}^{-1}$ by  $S$. The terms $a_{n}(s)$ in Theorem 6 of \cite{avron_adiabatic_2012}
satisfy $P(s)a_{n}(s)=a_{n}(s)$ while $Q(s)b_{n}(s)=b_{n}(s)$. The
initial condition implies that $\rho(0)=a_{0}(0)=\sigma(0)$ (and
$a_{n}(0)=0$ for $n\ge1$). The assumption that $\mathcal{L}(s)$
is TPHP together with uniqueness implies that $\dot{P}(s)b_{n}(s)=\dot{P}(s)Q(s)b_{n}(s)=0$.
In fact, using Hilbert-Schmidt scalar product notation, $P=|\sigma\rangle\langle\1|$,
so $\dot{P}=|\dot{\sigma}\rangle\langle\1|$ and $\dot{P}Q=0$ follows.
Since $S=SQ=QS$ we also have $\dot{P}S=S\dot{P}=0$. This implies, from Eq.~(15) of \cite{avron_adiabatic_2012}, that
that $a_{0}(s)=T(s,0)a_{0}(0)=\sigma(s)$ and $a_{n}(s)=0$ for $n\ge1$. Our Eq.~\eqref{eq:adia_series1} then follows from Eq.~(14) of \cite{avron_adiabatic_2012}. In addition, Eq.~\eqref{eq:remainder} is a special case of Eq.~(17) of \cite{avron_adiabatic_2012} under our additional initial condition assumption $\rho(0)=\sigma(0)$, which implies that $r_k(\zeta,0)=0$.
\end{proof}

Note that Eq.~\eqref{eq:adia_series1} can
be written as
\beq
\rho(s)=\sum_{n=0}^{k}\left(\zeta S(s)\frac{d}{ds}\right)^{n}\sigma(s)+\zeta^{k+1}r_{k} (\zeta,s) .
\eeq

\section{Explicit Constants for the Adiabatic Error\label{sec:Explicit-Constants-for}}

By the assumption of ergodicity $\dot{\sigma}=\dot{P}\sigma=-S\dot{\mathcal{L}}\sigma$.
We use some results from \cite{venuti_adiabaticity_2016}. We have
\bes
\begin{align}
b_{1} & =-S^{2}\dot{\mathcal{L}}\sigma\\
\dot{b}_{1} & =\left(2S^{2}\dot{\mathcal{L}}S\dot{\mathcal{L}}-2S^{3}\dot{\mathcal{L}}P\dot{\mathcal{L}}-S^{2}\ddot{\mathcal{L}}-P\dot{\mathcal{L}}S^{3}\dot{\mathcal{L}}+S\dot{\mathcal{L}}S^{2}\dot{\mathcal{L}}\right)\sigma\\
b_{2} & =S\left(2S^{2}\dot{\mathcal{L}}S\dot{\mathcal{L}}-2S^{3}\dot{\mathcal{L}}P\dot{\mathcal{L}}-S^{2}\ddot{\mathcal{L}}+S\dot{\mathcal{L}}S^{2}\dot{\mathcal{L}}\right)\sigma.
\end{align}
\ees
So 
\bes
\begin{align}
\left\Vert b_{1}\right\Vert  & \le\left\Vert S\right\Vert ^{2}\left\Vert \dot{\mathcal{L}}\right\Vert \\
\left\Vert \dot{b}_{1}\right\Vert  & \le6\left\Vert S\right\Vert ^{3}\left\Vert \dot{\mathcal{L}}\right\Vert ^{2}+\left\Vert S\right\Vert ^{2}\left\Vert \ddot{\mathcal{L}}\right\Vert \\
\left\Vert b_{2}\right\Vert  & \le5\left\Vert S\right\Vert ^{4}\left\Vert \dot{\mathcal{L}}\right\Vert ^{2}+\left\Vert S\right\Vert ^{3}\left\Vert \ddot{\mathcal{L}}\right\Vert .
\end{align}
\ees
Moreover, one can show that
\begin{align}
\left\Vert \dot{b}_{2}\right\Vert \le60\left\Vert S\right\Vert ^{5}\left\Vert \dot{\mathcal{L}}\right\Vert ^{3}+19\left\Vert S\right\Vert ^{4}\left\Vert \dot{\mathcal{L}}\right\Vert \left\Vert \ddot{\mathcal{L}}\right\Vert +\left\Vert S\right\Vert ^{3}\left\Vert \dddot{\mathcal{L}}\right\Vert .
\end{align}

Let us now use Proposition~\ref{prop:adiabatic_series} with $k=0$:
\begin{align}
\left\Vert \rho(1)-\sigma(1)\right\Vert  & \le\frac{1}{\tau}\left\Vert r_{0}(\tau,1)\right\Vert \notag \\
 & \le\frac{1}{\tau}\left(\left\Vert b_{1}(1)\right\Vert +\left\Vert b_{1}(0)\right\Vert +\sup_{s\in[0,1]}\left\Vert \dot{b}_{1}(s)\right\Vert \right)\notag \\
 & =\frac{B_{0}}{\tau} ,
\end{align}
where one can take $B_0$ as in Eq.~\eqref{eq:B0}.
Similarly we can use Proposition~\ref{prop:adiabatic_series} with $k=1$ and obtain: 
\begin{align}
\left\Vert \rho(1)-\sigma(1)\right\Vert \le\frac{1}{\tau}\left\Vert b_{1}(1)\right\Vert +\frac{1}{\tau^{2}}\left\Vert r_{1}(\tau,1)\right\Vert ,
\end{align}
implying
\begin{equation}
\left\Vert \rho(1)-\sigma(1)\right\Vert \le\frac{A_{1}}{\tau}+\frac{B_{1}}{\tau^{2}},\label{eq:bound_t1_t2}
\end{equation}
with $A_1$ and $B_1$ as in Eqs.~\eqref{eq:A1} and \eqref{eq:B1}, respectively.

\section{Derivation of the Schr\"{o}dinger picture and adiabatic Redfield master equations}
\label{app:Redfield}

\subsection{SPRME}

The first few steps are customary. In the interaction picture, after the Born approximation (see Eq.~(3.116) of \cite{breuer_theory_2007} and additional details therein) the dynamics of the system's density matrix in the interaction picture $\rho_{I}$ are given by
\beq
\dot{\rho}_{I}(t)=-\int_{0}^{t}dt'\Tr_{B}\left[H_{I}(t),\left[H_{I}(t'),\rho_{I}(t')\otimes\rho_{B}\right]\right].
\label{eq:Born}
\eeq
After substituting the interaction Hamiltonian $H_{I}=g\sum_{\alpha}A_{\alpha}\otimes B_{\alpha}$ and a change of integration variable, we obtain: 
\beq
\dot{\rho}_{I}(t)=\sum_{\alpha,\beta}\int_{0}^{t}dt' G_{\alpha,\beta}(t-t')\,\left[A_{\beta}(t')\rho_{I}(t'),A_{\alpha}(t)\right]+\hc ,\label{eq:pre_redfield}
\eeq
which is the same as Eq.~(9) of \cite{albash_quantum_2012}. 
This
equation is non-local in time because on the right-hand-side the unknown $\rho_{I}(t)$  appears also for times $t'\neq t$. In order to make it time-local
we use the Markov approximation $\rho_{I}(t')\approx\rho_{I}(t)$. When this substitution is made in Eq.~\eqref{eq:Born} the resulting equation is called the Redfield master equation (RME) according to \cite{breuer_theory_2007} [see Eq.~(3.117) therein], though in our case the system Hamiltonian is explicitly time-dependent, in contrast to standard Redfield theory.

We now estimate the error made with this Markov approximation. We use the same techniques utilized in Appendix B of Ref.~\cite{albash_quantum_2012} but note that we do not extend the upper integration limit to $\infty$ as done there. One can then show that the relative error of this approximation is of the order of magnitude of the following integral: 
\beq
\sum_{\alpha,\beta}\int_{0}^{\tau}dt'\,t' \vert G_{\alpha,\beta}(t') \vert.
\label{eq:C3}
\eeq
If $G_{\alpha,\beta}(t)$ is an exponentially decaying function of $t$
with time decay constant $\tau_{B}$ (``fast bath''), the integral in Eq.~\eqref{eq:C3} is of the order of  $O(\tau_{B}^{2}g^{2})$, in agreement with Ref.~\cite{albash_quantum_2012}.

Consider now the case where $G_{\alpha,\beta}(t)$ decays algebraically. I.e., assume that for times $t>t_{0}$, 
\beq
\vert G_{\alpha,\beta}(t) \vert\sim g^{2}\left(\frac{\tau_{M}}{t}\right)^{\theta},
\eeq
where we neglected the (unimportant) dependence on the labels $\alpha,\beta$. For $\theta\neq 2$ the relative error is then of the order of
\begin{equation}
 \int_{t_{0}}^{\tau}dt'\,t' \vert G_{\alpha,\beta}(t') \vert=\frac{g^{2}\tau_{M}^{\theta}}{\theta-2}\left(\frac{1}{t_{0}^{\theta-2}}-\frac{1}{\tau^{\theta-2}}\right).\label{eq:alpha_not2}
\end{equation}
For the case $\theta=2$ we obtain instead
\beq
\int_{t_{0}}^{\tau}dt'\,t' \vert G_{\alpha,\beta}(t') \vert\sim\left(g\tau_{M}\right)^{2}\ln\left(\tau/t_{0}\right).
\eeq
Note that for  $\theta<2$ the relative error  Eq.~\eqref{eq:alpha_not2} increases as $\tau$ grows larger and in fact diverges as $\tau\to\infty$. The same is true for $\theta=2$ although in this case the growth is only logarithmic. Keeping the upper integration limit in Eq.~\eqref{eq:pre_redfield}, finite and bounded by $\tau$,  circumvents this problem in case of an insufficiently fast bath. 

Neglecting this error, we obtain: 
\begin{equation}
\dot{\rho}_{I}(t)=\sum_{\alpha,\beta}\int_{0}^{t}dt' G_{\alpha,\beta}(t')\,\left[A_{\beta}(t-t')\rho_{I}(t),A_{\alpha}(t)\right]+\hc
\label{eq:red_almost}
\end{equation}
This is still the RME, in somewhat more explicit form. After transforming back to the Schr\"{o}dinger picture via $\rho(t)=U_{0}(t,0)\rho_{I}(t)U_{0}(0,t)$,
we directly obtain the SPRME given in Eq.~\eqref{eq:nonadia_red}.

\subsection{ARME}
At this point, since the bath correlation function is peaked in a small time-window 
we can expand $U_{0}(t,t-t')$ in powers of $t'$. This is effectively an expansion in powers of $t'/\tau$, and we can use
\begin{equation}
U_{0}(t,t-t')=e^{-i s H_{S}(t)}+O\left((t'/\tau)^{2}\right).
\label{eq:approx}
\end{equation}
If $G_{\alpha,\beta}(r)$ decays exponentially the error of this latter approximation is then of
order of 
\begin{equation}
\int_{0}^{\infty}dt'\,\left(\frac{t'}{\tau}\right)^{2}\left|G_{\alpha,\beta}(t')\right|=\tau_{B}^{3}\left(\frac{g}{\tau}\right)^{2},
\end{equation}
while the leading term [Eq.~\eqref{eq:red_almost} after substituting Eq.~\eqref{eq:approx}] is $O\left(\tau_{B}g^{2}\right)$. Dividing, the
relative error is 
\beq
\left(\frac{\tau_{B}}{\tau}\right)^{2}\ll1.
\eeq
In this sense this approximation is adiabatic as it requires $\tau$ large, i.e., $\tau\gg\tau_{B}$.

For an algebraic bath, the order of magnitude of the (absolute) error is:
\begin{equation}
\int_{t_{0}}^{\tau}dt'\,\left(\frac{t'}{\tau}\right)^{2}\left|G_{\alpha,\beta}(t')\right|=\frac{g^{2}\tau_{M}^{\vartheta}}{\tau^{2}(\vartheta-3)}\left(t_{0}^{-(\vartheta-3)}-\tau^{-(\vartheta-3)}\right),
\end{equation}
while the order of magnitude of the leading term is:
\begin{equation}
\int_{t_{0}}^{\tau}dt'\,\left|G_{\alpha,\beta}(t')\right|=\frac{g^{2}\tau_{M}^{\vartheta}}{\tau^{2}(\vartheta-1)}\left(t_{0}^{-(\vartheta-1)}-\tau^{-(\vartheta-1)}\right).
\end{equation}
For $\vartheta>3$ the relative error becomes,  assuming $\tau\gg t_{0}$,
$O\left(\left(t_{0}/\tau\right)^{2}\right)$ and is small in the
adiabatic limit (here $\tau\gg t_{0}$). Instead, e.g., for $\vartheta=2$ one
obtains $O\left(\tau/t_{0}\right)$ and so the error is large. 

Finally, discarding the error term in Eq.~\eqref{eq:approx} we obtain the ARME given in Eq.~\eqref{eq:adia_red}. 


\section{Numerical computation of the integral in Eq.~\eqref{eq:nonadia_red}} 
\label{app:X}

We consider the case of a single system-bath operator $A$; generalization
is straightforward. The integral that we need to compute is
\beq
W(t)=\int_{0}^{t}dr\,G(r)U_{0}(t,t-r)AU_{0}(t-r,t).
\label{eq:W}
\eeq

Recall that the standard fourth order Runge-Kutta, which is routinely
used in many ODE solvers, is equivalent to Simpson's rule for integration.
With this in mind we simply implement the integral using Simpson's
rule. This means that the integral in Eq.~\eqref{eq:W} is replaced by the following
sum:
\beq
\int_{0}^{t}drf(r)\approx\Delta r\sum_{j=0}^{n}f(r_{j})w_{j},
\eeq
where $r_{j}=j\Delta r$, for $j=0,1,\ldots,n$, $\Delta r=t/n$,
and $w_{0}=w_{n}=1/3$ while $w_{j}=4/3$ for $j$ odd and $w_{j}=2/3$
for $j$ even. The error in the Simpson's rule is $\propto t(\Delta r)^{4}=t^{5}n^{-4}$.
In order to have a constant error we must pick $n\propto t^{5/4}=t^{1.25}$.
To be conservative, in our simulations we pick $n\propto t^{1.3}$.
Next we need to compute $T_{j}:=U_{0}(t,t-r_{j})$. We use
\begin{align}
T_{0} & =U_{0}(t,t)=\1\notag\\
T_{1} & =U_{0}(t,t-\Delta r)\approx\exp\left(-i\Delta rH_{S}(t)\right)\\
T_{j+1} & :=T_{j}\exp\left(-i\Delta rH_{S}(t-r_{j})\right),\,\,j=0,1,\ldots,n-1. \notag
\end{align}
Note that this approximation preserves unitarity, i.e., $T_{j}T_{j}^{\dagger}=\1$
for $j=0,1,\ldots,n$. 

Finally, the differential equation \eqref{eq:nonadia_red} assumes the form
\begin{align}
\frac{\partial\rho(t)}{\partial t}  & =-i\left[H_{S}(t),\rho(t)\right] +\left(\left[W(t)\rho(t),A\right]+\hc\right).
\end{align}


%

\end{document}